\documentclass[12pt,reqno]{amsart}

\usepackage{latexsym}
\usepackage{amsmath}
\usepackage{amscd}
\usepackage{amssymb}
\usepackage{dsfont}

\usepackage{mathrsfs}
\usepackage{comment}
\usepackage{color}
\usepackage{enumerate}
\usepackage{soul}
\usepackage{framed}

\usepackage{graphicx}
\usepackage{subfigure}
\usepackage{url}

\newtheorem{theorem}{Theorem}[section]

\newtheorem{lemma}[theorem]{Lemma}

\newtheorem{remark}[theorem]{Remark}

\numberwithin{equation}{section}
\newcommand{\bfi}{\bfseries\itshape}

\newcommand{\rem}[1]{}

\begin{document}

\title{ $G$-Strands on symmetric spaces}

\author[
A. Arnaudon, D. Holm and R Ivanov  ]
{Alexis Arnaudon$^{1}$, Darryl D. Holm$^{1}$ and Rossen I. Ivanov$^{2}$  }

\address{$^{1}$ Department of Mathematics, Imperial College, London SW7 2AZ, UK.\\
$^{2}$ School of Mathematical Sciences, Dublin Institute of Technology, Kevin Street, Dublin 8, Ireland.}

\begin{abstract}
    We study the $G$-strand equations that are extensions of the classical chiral model of particle physics in the particular setting of broken symmetries described by symmetric spaces. 
    These equations are simple field theory models whose configuration space is a Lie group, or in this case a symmetric space. 
    In this class of systems, we derive several models that are completely integrable on finite dimensional Lie group $G$ and we treat in more details examples with symmetric space $SU(2)/S^1$ and $SO(4)/SO(3)$.  
    The later model simplifies to an apparently new integrable $9$ dimensional system.  
    We also study the $G$-strands on the infinite dimensional group of diffeomorphisms, which gives, together with the Sobolev norm, systems of $1+2$ Camassa-Holm equations. 
    The solutions of these equations on the complementary space related to the Witt algebra decomposition are the odd function solutions. 
\end{abstract}

\maketitle

\section{Introduction}
We study a simplified field theoretical model called $G$-strands in which the fields take values in a Lie group $G$. 
The $G$-strands are related to the chiral model of particle physics, and they have been studied in the context of geometric mechanics in \cite{ Ho-Iv-Pe,Ho-Iv1,Ho-Iv2, FDT, HoLu2013}. 
In this paper, we are considering $G$-strands where the Lie group is replaced by a symmetric space.
The passage from a homogeneous space (that is, the quotient space of a Lie group by one of its subgroups) to a symmetric space requires an involutive automorphism which provides more structure for the equations.
This sort of geometrical structure appears in the theory of complex fluids, which is based on the concept of order parameters resulting from broken symmetries.
More precisely, the order parameter belongs to the coset space of the broken symmetry with respect to the remaining symmetry, which is a homogeneous space.

This definition of order parameter includes liquid crystals, He$^3$, He$^4$ and their various generalisations in condensed matter theory.
Here we consider models of strands on symmetric spaces for which the dynamics will be shown to be completely integrable. 
Some of these results are also compatible with well-known chiral models, see for example \cite{ZaMi}.

In our previous publications \cite{Ho-Iv-Pe,Ho-Iv1, Ho-Iv2,FDT, HoLu2013} we introduced and studied the $G$-strand construction, which gives rise to equations for a map $\mathbb{R}\times \mathbb{R}$ into a Lie group $G$ associated to a $G$-invariant Lagrangian. In the case of a semisimple Lie group $G$ with a Lie algebra $\frak{g}$, various classes of integrable equations have been found. In these cases the Lax operators take values in a loop algebra $\frak{g}_{\lambda}=\mathbb{C}[\lambda,\lambda^{-1}]\otimes \frak{g}$, and the resulting equations correspond to the classes of possible loop algebras $\frak{g}_{\lambda}$.
Other possibilities for the derivation of non-equivalent integrable systems use the concept of automorphic Lie algebras, which are subalgebras of $\frak{g}_{\lambda}$. Such subalgebras are obtained as a reduction with respect to an automorphism $\varphi$ of  $\frak{g}_{\lambda}$, i.e. they are the $\varphi$-invariant part of $\frak{g}_{\lambda}$, see \cite{MiLo}. The set of all automorphisms of $\frak{g}_{\lambda}$ forms the reduction group, as introduced in \cite{Mi}. Knowledge of the reduction group is important for the classification of the integrable equations as reductions from a given $\frak{g}_{\lambda}$. The reduction group naturally acts on all structures related to the Lax operator, including the scattering data. Both the continuous and discrete spectra of the Lax operator are orbits of the reduction group, see \cite{Mi}. Moreover, the automorphisms act naturally on the phase space and Hamiltonian structures, thereby introducing reductions to them as well.

Here we shall address a different but related reduction of an integrable system associated to $\frak{g}_{\lambda}$ which makes use of symmetric spaces. The structure of a symmetric space is determined by an involutive automorphism of the Lie algebra $\frak{g}$, known as Cartan involution, and the corresponding decomposition 
\begin{align*}
  \mathfrak g=\mathfrak  k\oplus\mathfrak p\, ,
\end{align*}
  where $ \mathfrak  k$ is a subalgebra, invariant under the Cartan involution, and $\mathfrak p$ is a complementary subspace on which the Cartan involution has an eigenvalue $-1$. 
The classification of the symmetric spaces of the simple Lie groups is provided in the classic monograph \cite{h}. 
Due to the Lie-algebraic nature of this splitting, the Hamiltonian variables of these equations separate into sets taking values in either $\mathfrak  k$ or $\mathfrak p$. 
Moreover, by restricting the Hamiltonian to depend only on the variables in the space $\mathfrak p$, the reduced equations can be written on the symmetric space.  
We will explain how this construction is related to the concept of un-reduction of \cite{BrGBHoRa2011} which was extended to covariant field theories in \cite{ACH2015}.
Integrable systems on symmetric spaces of finite dimensional Lie algebras have been well studied in the literature, see \cite{AtFo, KuFo, Fo, GeGr,GGK}.  
A more general but similar construction would be on homogeneous spaces, as studied recently in \cite{Vi1, Vi2} also in the context of reduction by symmetry.
This construction is different from the semidirect product $G$-strands explored in \cite{Ho-Iv2} for the special Euclidean group $SE(3):= SO(3) \circledS \mathbb R^3$ with Lie algebra commutation relations  
\begin{align*}
        [\mathfrak{so}(3), \mathfrak{so}(3)] \subset \mathfrak{so}(3)\quad  \mathrm{and}  \quad [\mathfrak{so}(3),\mathbb R^3] \subset  \mathbb R^3\, ,
\end{align*}
because for symmetric spaces, an additional commutation relation occurs involving the complementary space $\mathfrak{p}$ 
\begin{align*}
        [\mathfrak{k}, \mathfrak{k}] \subset \mathfrak{k}\, , \qquad [\mathfrak{k},\mathfrak{p} ] \subset  \mathfrak{p}\, , \qquad [\mathfrak{p},\mathfrak{p}] \subset \mathfrak{k}\, .
\end{align*}

{\bf Plan of the work.} In Section \ref{theory-review} we give a brief account of the $G$-strand construction, reviewing our previous works and illustrating the $G$-strand construction with several simple but instructive examples. 
In Section \ref{symmetric-spaces} we construct the theory of integrable $G$-strands on symmetric spaces for semisimple Lie algebras $\frak{g}$, and give examples for $SU(2)$ and $SO(4)$. 
The idea of the restriction of the phase space of a Hamiltonian system on a symmetric space may also be useful for non-integrable systems. For example, non-integrable $G$-strands arise when $G$ is the infinite-dimensional group of diffeomorphisms. $\mathrm{Diff}(\mathbb{R})$-strand equations on symmetric spaces and their singular solutions are presented in Section \ref{diff-strand}. 

\section{The $G$-strand equations}\label{theory-review}

We describe here the construction of the $G$-strand equations based on the theory of reduction by symmetry for covariant field theory in $1+1$ dimensions. 
More general systems which encompass the $G$-strand equations can be found, for example, in \cite{CaRa}.

Let $G$ be a Lie group and consider the map $g(t,s): \mathbb{R}\times
\mathbb{R}\to G$. 
This map has two tangent vectors associated to the independent variables $s$ and $t$. 
We will denote them as $\dot{g} :=\frac{\partial g}{\partial t}\in T_g G$ and $g' :=\frac{\partial g}{\partial s} \in T_g G$, respectively. 
Although, as we will see later,  the dynamical equations will be symmetric under $s\leftrightarrow t$,  the time derivative can still be interpreted as the velocity, and the space derivative as a deformation gradient. 

We will now implement the theory of reduction by symmetry, which requires the Lagrangian
density function $ L(g,\dot{g},g') $ to be left $G$--invariant and thereby yields a reduced Lagrangian $l: \mathfrak{g}\times \mathfrak{g} \rightarrow \mathbb{R}$, defined by
\begin{align*}
  l(g^{-1}\dot{g} , g^{-1} g' ):= L(g^{-1}g,g^{-1}\dot{g},g^{-1}g')= L(g,\dot{g},g')\, .
\end{align*}
Conversely, this relation defines for any reduced Lagrangian
$l=l({\sf u},{\sf v}) : \mathfrak{g}\times \mathfrak{g}
\rightarrow \mathbb{R} $ a left $G$-invariant function $ L : T
G\times TG \rightarrow \mathbb{R} $ and a map $g(t,s):
\mathbb{R}\times \mathbb{R}\to G$ such that
\begin{align}
    {\sf u} (t,s) := g^{ -1} g_t (t,s) =:g^{ -1}\dot{g}(t,s)
    \quad\hbox{and}\quad {\sf v} (t,s) := g^{ -1} g_s (t,s)=: g^{ -1}g' (t,s)\, .
    \label{reduced-vectors}
\end{align}

\begin{theorem} [Covariant Euler-Poincar\'e theorem]\label{lall}
With the preceding notation, the following two statements are equivalent:
\begin{enumerate}
    \item  The variational principle 
    \begin{align*}
        \delta \int _{t_1} ^{t_2} L(g(t,s), \dot{g} (t,s), g'(t,s) ) \,ds\,dt = 0 
    \end{align*}
    holds on $T G\times TG$, for variations $\delta g(t,s)$ of $ g (t,s)$ that vanish at the endpoints in $t$ and $s$. The function $g(t,s)$
    satisfies the Euler--Lagrange equation for $L$ on $G$, given by
    \begin{align} \label{EL-eqns}
        \frac{\partial L}{\partial g} - \frac{\partial}{\partial t}\frac{\partial L}{\partial g_t} - \frac{\partial}{\partial s}\frac{\partial L}{\partial g_s} = 0\, .
    \end{align}
    \item  The constrained variational principle
    \begin{align} \label{variationalprinciple}
        \delta \int _{t_1} ^{t_2}  l({\sf u}(t,s), {\sf v}(t,s)) \,ds\,dt = 0
    \end{align}
    holds on $\mathfrak{g}\times\mathfrak{g}$, using variations of ${\sf u} := g^{ -1} g_t (t,s)$ and ${\sf v}:= g^{ -1} g_s(t,s) $ of the forms
    \begin{align} \label{epvariations}
        \delta {\sf u} = \dot{{\sf w} } + {\rm ad}_{\sf u}{\sf w} \quad\hbox{and}\quad \delta {\sf v} = {\sf w}\,' + {\rm ad}_{\sf v} {\sf w} \, ,
    \end{align}
    where ${\sf w}(t,s) :=g^{ -1}\delta g \in \mathfrak{g}$ vanishes
    at the endpoints. The {\bfi Euler--Poincar\'{e}} equation holds on
    $\mathfrak{g}^*\times\mathfrak{g}^*$
    \begin{align}
    \frac{d}{dt} \frac{\delta l}{\delta {\sf u}} -
     \operatorname{ad}_{{\sf u}}^{\ast} \frac{ \delta l }{ \delta {\sf u}}
    + \frac{d}{ds} \frac{\delta l}{\delta {\sf v}} -
     \operatorname{ad}_{{\sf v}}^{\ast} \frac{ \delta l }{ \delta {\sf v}} = 0\, ,  
    \label{EPeqns}
    \end{align}
where ${\rm ad}^*: \mathfrak{g}\times\mathfrak{g}^*\to\mathfrak{g}^*$ is defined via
${\rm ad}:\mathfrak{g}\times\mathfrak{g}\to\mathfrak{g}$ with the pairing of the Lie algebra $\langle \,\cdot\,,\,\cdot\,\rangle: \mathfrak{g}^*\times\mathfrak{g}\to\mathbb{R}$.
\end{enumerate}
\end{theorem}
The proof of this theorem is classical and can be found in previous works, \cite{MaRa1999,CaRa}.
The reduced fields ${\sf u}$ and ${\sf v}$ have an additional property in the context of covariant field theory which is given in the following lemma. 
\begin{lemma}
The left-invariant tangent vectors ${\sf u} (t,s)$ and ${\sf v} (t,s)$ at the identity of $G$ satisfy {\bfi zero-curvature relation}
\begin{align}
    {\sf v}_t - {\sf u}_s = -\,{\rm ad}_{\sf u}{\sf v} \,.
    \label{zero-curv1}
\end{align}
\end{lemma}
\begin{proof}
    The proof is standard and follows from equality of cross derivatives $g_{ts}=g_{st}$.
\end{proof}
This solution of the Euler-Poincar\'e equation must satisfy this additional relation \eqref{zero-curv1} for it to correspond to a solution on the Lie group, reconstructed via the definition of the reduced vectors \eqref{reduced-vectors}. 
As we will work with the reduced variables and also need the reconstruction to hold, we will always impose this relation and solve the coupled system of Euler-Poincar\'e equation \eqref{EPeqns} and ZCR \eqref{zero-curv1} together. 
We will call these equations the {\bfi $G$-strand equations}. 

\begin{remark}[Zero curvature relation]
The term zero curvature relation is not related to the curvature of the space where the fields $\sf u$ and $\sf v$ take values. Rather, it refers  to the curvature of another geometrical construction which will briefly describe here. 
The fields $\sf u$  and $\sf v$ are in fact maps ${\sf u}:\mathbb R^2\to \mathfrak g$. These maps combine into a single one-form $\nu := {\sf u}dt + {\sf v} ds$, which is a section of the bundle $\mathfrak g\otimes \mathbb R^2\to \mathbb R^2$, mapping from the space-time $(t,s)\in \mathbb R^2$ to one-forms on $\mathbb R^2$ with values in $\mathfrak g$. 
The zero curvature relation is the condition that the curvature of this structure vanishes on the solution of the Euler-Poincar\'e equation. 
The curvature is expressed in term of the covariance exterior derivative which is here $\mathrm{Curv}(\nu) := d^\nu\nu = d\nu + [\nu,\nu]$. 
\end{remark}

\begin{remark}[Historical remark] 
In 1901, Poincar\'e \cite{Po1901} proved that when a Lie
algebra acts locally transitively on the configuration space of a
Lagrangian mechanical system, the well known Euler-Lagrange
equations are equivalent to a new system of differential equations
defined on the product of the configuration space with the Lie
algebra. These equations are now called Euler-Poincar\'e equations in his honour. In modern language the content of the
Poincar\'e's article \cite{Po1901} is presented for example in
\cite{Ho2011GM2,Marle}. An English translation of the article
\cite{Po1901} can be found in Appendix D of \cite{Ho2011GM2}.
\end{remark}

\subsection{Lie-Poisson Hamiltonian formulation}

In this section, we derive the corresponding $G$-strand equations in the Hamiltonian framework, by applying the Legendre transformation to the Lagrangian $\ell({{{\sf u}},{{\sf
v}}}):\, \mathfrak{g}\times \mathfrak{g}\to\mathbb{R}$ with respect to the first variable only. 
This yields the Hamiltonian $h({{{\sf m}},{{\sf v}}}):\, \mathfrak{g}^*\times 
\mathfrak{g}\to\mathbb{R}$, given by 
\begin{align}
    h({\sf m},{\sf v}) = \langle{\sf m}\,,\,{\sf u}\rangle - \ell({\sf u},{\sf v}) \, .
    \label{leglagham} 
\end{align}
The variational derivatives of the Hamiltonian gives the useful relations 
\begin{align*}
    \frac{\delta l}{\delta {\sf u}} = {{\sf m}} \,,\quad
    \frac{\delta h}{\delta {\sf m}} = {{\sf u}} \quad\hbox{and}\quad
    \frac{\delta h}{\delta {\sf v}} = -\,\frac{\delta \ell}{\delta {\sf v}} = -{\sf n}\, .
\end{align*}
The corresponding non-canonical Hamiltonian equation, or Lie-Poisson equations are directly found to be 
\begin{align} \label{hameqns-so3}
\begin{split}
    {\partial_t} {{\sf m}} &= {\rm ad}^*_{\delta h/\delta {{\sf m}}}\, {{\sf m}} + \partial_{s} \frac{\delta h}{\delta {{\sf v}}} - {\rm ad}^*_{{\sf v}}\,\frac{\delta h}{\delta {{\sf v}}} \, ,\\
\partial_t {{\sf v}}
    &= \partial_{s}\frac{\delta h}{\delta {{\sf m}}} -  {\rm ad}_{\delta h/\delta {{\sf m}}}\,{{\sf v}} \,.
\end{split}
\end{align}
Assembling these equations into Lie-Poisson Hamiltonian structure 
gives
\begin{align} \label{LP-Ham-struct-symbol1}
\frac{\partial}{\partial t}
    \begin{bmatrix}
    {{\sf m}}
    \\
    {{\sf v}}
    \end{bmatrix}
=
\begin{bmatrix}
  {\rm ad}^*_\square {{\sf m}}
   &\hspace{5mm}
  \partial_s - {\rm ad}^*_{{\sf v}}
   \\
   \partial_s + {\rm ad}_{{\sf v}}
   &\hspace{5mm} 0
    \end{bmatrix}
    \begin{bmatrix}
            \delta h/\delta{{\sf m}}={\sf u} \\
            \delta h/\delta{{\sf v}}=-{\sf n}
    \end{bmatrix}\, ,
\end{align}
where we use the notation $(\mathrm{ad}^*_\square m) u = \mathrm{ad}^*_u m$. 

The Hamiltonian matrix in equation \eqref{LP-Ham-struct-symbol1}
also appears in the Lie-Poisson brackets for Yang-Mills plasmas,
for spin glasses and for perfect complex fluids, such as liquid
crystals, see for example \cite{HoKu1983,Ho2002,FGBRa2009}.

\begin{remark}[On the loss of covariance]
    The choice of taking the Legendre transform only with respect to the velocity ${\sf u}$, destroys the symmetry or covariance of the $t$ and $s$ variables in the Lie-Poisson equation. 
There exists an intrinsic way to apply the Legendre transformation that preserves this symmetry, developed in \cite{CaMa2003}, but we will not use this method here, as we want to relate the $G$-strand equations to classical Hamiltonian systems such as $\sigma$-models.
\end{remark}

\subsection{$G$-strand equations on semisimple Lie algebras}

Denoting ${\sf m}:=\delta \ell/\delta{\sf u}\in \mathfrak g^*$ and ${\sf n}:=\delta \ell/\delta{\sf v}\in \mathfrak{g}^*$, the $G$-strand equations are
\begin{align}
    \begin{split}
    {\sf m}_t + {\sf n}_{s} - {\rm ad}^*_{\sf u}{\sf m}
     - {\rm ad}^*_{\sf v}{\sf n} &=0\, ,\\
    \partial_t{\sf v} -\partial_{s}{\sf u} + {\rm ad}_{\sf u}{\sf v} &=0\, .
    \end{split}
\end{align}
For a semisimple \emph{matrix Lie group} $G$ and its \emph{semisimple Lie algebra} $\mathfrak{g}$, one has ${\rm ad}^*=-\,{\rm ad}$ and these equations take the commutator form, 
\begin{align} \label{MatrAlgEq} 
        \begin{split} {\sf m}_t + {\sf n}_{s} + {\rm ad}_{\sf u}{\sf m}
 + {\rm ad}_{\sf v}{\sf n} =&0\, , \\
\partial_t{\sf v} -\partial_{s}{\sf u}  + {\rm ad}_{\sf u}{\sf v}=& 0\, ,
\end{split}
\end{align}
where we have used the ad-invariant pairing of semisimple matrix Lie algebras which is given by the Killing form
\begin{align}
        \langle{\sf m}\,,\,{{{\sf n}}}\rangle=\mathrm{Tr}(\mathrm{ad}_{\sf m} \mathrm{ad}_ {\sf n}) = \epsilon\, \mathrm{Tr}({\sf m}{\sf n})\, , 
\end{align}
where $\epsilon$ is a negative constant which depends on the Lie algebra.
For example we have $\epsilon = -1/ 2$ for $\mathfrak{so}(3)$. 
This pairing is non-degenerate if the Lie algebra is semi-simple and thus allows us to identify $\mathfrak g\cong \mathfrak g^*$ and the adjoint operator which is the matrix commutator and is identified with minus the coadjoint operator.
Accordingly, the Hamiltonian structure reduces to 
\begin{align} \label{LP-Ham-struct-semisimple}
\frac{\partial}{\partial t}
    \begin{bmatrix}
    {{\sf m}}
    \\
    {{\sf v}}
    \end{bmatrix}
=
\begin{bmatrix}
  {\rm ad}_{\sf m}
   &\hspace{5mm}
  \partial_s + {\rm ad}_{{\sf v}}
   \\
   \partial_s + {\rm ad}_{{\sf v}}
   &\hspace{5mm} 0
    \end{bmatrix}
    \begin{bmatrix}
            \delta h/\delta{{\sf m}}={\sf u} \\
            \delta h/\delta{{\sf v}}=-{\sf n}
    \end{bmatrix}\, .
\end{align}
Examples of these systems for various Lie groups are studied in \cite{Ho-Iv-Pe, Ho-Iv2, FDT}.

\subsection{Example: the chiral model}

For the fields ${\sf m}$ and ${\sf v}$ with values in a Lie algebra $\mathfrak g$, we choose the quadratic Hamiltonian density 
\begin{align}
    h({\sf m},{\sf v}) = \frac{1}{2} \langle{\sf m}\,,\,{\sf m}\rangle +\frac12  \langle{\sf v}\,,\,{\sf v}\rangle \,.
\label{hsu2} 
\end{align} 
We thus have ${\sf u}={\sf m}$ and ${\sf v}={-\sf n}$ and the Hamiltonian equations corresponds to the well known chiral model for $\mathfrak{su}(n)$
\begin{align} \label{suNchiral} 
\begin{split} 
        {\sf v}_{s} -{\sf u}_t  &=0\, ,\\
    {\sf v}_t -{\sf u}_s  + [{\sf u},{\sf v}]&= 0\, ,
\end{split}
\end{align}
see for example \cite{ZaMi, FaTa, ZMNP, GVY} and references therein.
This is an integrable model with a Lax representation 
\begin{align}
    L_t-M_s+[L,M]=0\, ,
\end{align}
where 
\begin{align}
    L&=\frac{1}{4}\left(- 2{\sf v}  +\lambda({\sf u} - {\sf v})-\frac{1}{\lambda}({\sf u}+{\sf v})\right)   \\
    M&=-\frac{1}{4}\left( 2{\sf u}  +\lambda({\sf u} - {\sf v})+\frac{1}{\lambda}({\sf u}+{\sf v})\right)\, ,
\label{LM-chiral}
\end{align}
for an arbitrary complex spectral parameter $\lambda$. 
This equivalent formulation of the equation \eqref{suNchiral} makes this model integrable by the means of the inverse scattering method.
In addition, the chiral model equations \eqref{suNchiral} can be rewritten in a more familiar form, first derived in \cite{ZaMi}, by doing the following.
First apply the change of space time-variables  $t= \frac12(x-y)$ and $s= \frac12 (x+y)$ as well as new fields $\xi= u-v$ and $\eta= u+v$. 
The chiral model \eqref{suNchiral} transforms accordingly to 
\begin{align}
    \begin{split}
    \partial_x \xi  -\frac12[\xi,\eta]&= 0 \\
    \partial_y\eta + \frac12 [\xi,\eta]&= 0\, .
    \end{split}
\end{align}
The Lax representation \eqref{LM-chiral} is also modified to
\begin{align}
        L&= \frac{\xi}{\lambda-1}\quad \mathrm{and} \quad M= \frac{\eta}{\lambda+1}\, .
\end{align}
The choice of $SU(2)$ is of particular interest because after a change of variables shown in \cite{ZMNP} the equations result in a generalization of the sine-Gordon equation.  
The solutions of the chiral model on $SU(n)$ and $SO(n)$ are also discussed in \cite{ZMNP} as well as the cases of $U(n)$ and $SL(n)$ in \cite{Harnad1984,FaTa}.
For the more general case of $GL(n)$ we refer to \cite{Beggs1990}.

\section{Symmetric spaces for semisimple Lie algebras}\label{symmetric-spaces}

We will now formulate the $G$-strand equations on symmetric spaces.
To begin, we recall the definition of a symmetric space and refer the interested reader to the monographs \cite{h,arvanitogeorgos2003introduction} for more details. 

A homogeneous space is a manifold $\mathcal{M}$ on which a Lie group $G$ acts transitively. 
As a consequence, $\mathcal{M}$ is diffeomorphic to the coset space $G/K$, where $K$ is a (closed) Lie subgroup of $G$. 
Furthermore, in an important special case, the homogeneous space is reductive and its tangent space at the identity can be identified with a subspace $\mathfrak{p}$ of the Lie algebra $\mathfrak{g}$ of $G$. 
A large class of homogeneous spaces have the special geometrical properties which makes them symmetric spaces. 
This is the case when $K\subset G$ is also a subgroup of 
\begin{align}
        G^\varphi= \{ g\in G| \varphi(g)= g\}\, ,
\end{align}
where $\varphi:G\to G$ is an involution, i.e $\varphi^2(g)= g$. 
The involution $\varphi$ has an induced action $\tilde{\varphi}$ on $\mathfrak{g}$ and  
\begin{align}
        \mathfrak{k}=\{ X\in \mathfrak{g} , \tilde{\varphi}(X)=X\} \qquad \mathfrak{p}=\{ X\in \mathfrak{g} , \tilde{\varphi}(X)=-X\} \, , 
\end{align}
with
\begin{align}
        \mathfrak{g} = \mathfrak{k}\oplus \mathfrak{p}\, ,
\end{align}
where $\mathfrak{k}$ is a subalgebra, invariant under the Cartan involution and corresponding to the eigenvalue $+1$ f $\varphi$, and $\mathfrak{p}$ is a complimentary subspace on which the Cartan involution has an eigenvalue $-1$. 
The orthogonality between $\mathfrak{k}$ and $\mathfrak{p}$ is with respect to the Killing form of $\mathfrak{g}$ and the subspace $\mathfrak{k}$ is $Ad(K)$-invariant. 
Moreover, the following relations are fulfilled
\begin{align}
        [\mathfrak{k}, \mathfrak{k}] \subset \mathfrak{k}\, , \qquad [\mathfrak{k},\mathfrak{p} ] \subset  \mathfrak{p}\, , \qquad [\mathfrak{p},\mathfrak{p}] \subset \mathfrak{k}\, .
    \label{sym-relation}
\end{align}
The first relation means that $\mathfrak k$ is a Lie subalgebra, the second that $\mathfrak p$ is invariant under the action of $\mathfrak k$ and the third is a characteristic of symmetric spaces which distinguishes them from the semidirect product systems. 
We refer to \cite{h} for a complete classification of symmetric spaces for matrix Lie groups. 

\subsection{Reduction and un-reduction}\label{un-reduction}

Before going into the derivation of the $G$-strand equations on symmetric spaces, we want to highlight the underlying geometry associated with group reduction in the symmetric space construction. 
In this construction, we will select a Lagrangian, or Hamiltonian that is invariant with respect to the action of the full group $G$, but we ultimately want to have a system written on $T_eP= T_e(G/K)$ and not on $\mathfrak g$.  
For this, there is an interesting general construction based on the reduction by symmetry which can be applied directly here. 
For simplicity, we will only consider the classical mechanical setting, namely with no space $s$ variable. 
This construction was used in \cite{bruveris2011reduction} in the context of image matching and extended to field theories in \cite{ACH2015}.
The idea is to combine the un-reduction scheme of \cite{bruveris2011reduction} with the usual reduction by symmetry to obtain a dynamical equation on a Lie algebra rather than on the tangent bundle of a symmetric space. 
This can be achieved only for a particular class of Lagrangians which are symmetric with respect to both groups involved in the construction of the symmetric space and which do not depend on the complementary subspace in the Cartan decomposition.  
Specifically, we apply the following procedure of un-reduction and reduction. 
\begin{enumerate}
        \item This scheme works when the original system is described by a Lagrangian defined on the tangent space of a symmetric space $P:=G/K$; namely, $$\mathcal L: TP\to \mathbb R\, .$$ 
        We assume that this Lagrangian is invariant under the Lie groups $G$ and $K$. 
        For such a system, we cannot apply any reduction by symmetry for this Lagrangian, as $TP$ is not the tangent space of a Lie group.  
    \item To overcome this difficulty we append to $TP$ the so-called adjoint bundle $\tilde {\mathfrak{k}}:= (G\times \mathfrak k)/ K$, where the quotient is taken with respect to the group action of $K$ on $G$ and the adjoint action of $K$ on $\mathfrak k$. 
        The Lagrangian $\mathcal L$ can then be trivially extended to a Lagrangian on this space which does not explicitly depend on the variable in the adjoint bundle $\tilde{\mathfrak k}$. 
        We thus have an equivalent system described by the Lagrangian $$\overline L: TP\oplus \widetilde{\mathfrak k}\to \mathbb R\, .$$
    \item This extension of the original phase space $TP$ allows us to use a more general theory of reduction by symmetry, called Lagrange-Poincar\'e reduction theory \cite{cendra2001lagrangian}. 
        Presenting this theory in detail is out of the scope of this work; so we will just explain its main ideas. 
        First, this theory can be applied to general Lagrangian systems invariant under a Lie group whose dimension is smaller than the dimension of the configuration manifold, which is the manifold $M$ if the Lagrangian is written on $TM$. 
        Here, the Lie group is $K$ and it acts on a larger configuration manifold (which happens to be a Lie group) $G$.
        Second, the most important tool in this theory is the isomorphism $\alpha: T(G/K) \to TP\oplus \widetilde{\mathfrak k}$. This isomorphism is used to define another Lagrangian $l$ on $T(G/K)$ which is equivalent to $\overline L$, that is  $$l:(TG)/K\to \mathbb R\, . $$
        \item This step is the last one in the un-reduction procedure which uses the Lagrange-Poincar\'e reduction in the reverse direction to obtain the equivalent system  on $TG$, with corresponding Lagrangian $$L: TG\to \mathbb R\, .$$ 
        This step is described in detail in \cite{cendra2001lagrangian,bruveris2011reduction,ACH2015}. 
        \item The Lagrangian $L$ in the previous step is still equivalent to the original Lagrangian $\mathcal L$ which was invariant under the action of $G$. 
        We can thus use the standard Euler-Poincar\'e reduction theory to reduce this last system with Lagrangian $L$ to a Euler-Poincar\'e system with reduced Lagrangian  $$\ell:\mathfrak g\to \mathbb R\, .$$
\end{enumerate}

The crucial property of this last Lagrangian $\ell$ is that it will not depend on the $\mathfrak k$, as none of the previous Lagrangians did, but the equation of motion will involve $\mathfrak k$, as the Euler-Poincar\'e reduction is done using the full group $G$. 
We will use this fact later, starting with a Lagrangian whose resulting dynamical system can be written on a symmetric space. 

\subsection{$G$-strand equations}

In this section, we will derive the $G$-strand equations. We start from the complete Lie algebra, then apply the symmetric space definition to show that the equations reduce as expected from the previous theoretical considerations. 
We first split the variable $\sf m$ and the Hamiltonian according to \eqref{sym-relation} as
\begin{align}
    {\sf m}&=({\sf m}_-,{\sf m}_+) \in (\mathfrak  k,\mathfrak  p) \,,\\
     \frac{\delta H}{\delta {\sf m}}&=\left ( \frac{\delta H}{\delta {\sf m}_-},
     \frac{\delta H}{\delta {\sf m}_+}\right )
      \in (\mathfrak k,\mathfrak p)\, . 
\label{decomp} 
\end{align}
Note that ${\sf m} _{\pm}$ still belong to $\mathfrak g$ and since $\mathfrak  k$ and $\mathfrak  p$ are mutually orthogonal we have  the direct sum decompositions ${\sf m}={\sf m_-}+{\sf m_+}$ and  
\begin{align*}
    \frac{\delta H}{\delta {\sf m}}= \frac{\delta H}{\delta {\sf m}_-}+ \frac{\delta H}{\delta {\sf m}_+}\, .
\end{align*}

The Lie-Poisson Hamiltonian structure of the $G$-strand equation \eqref{LP-Ham-struct-symbol1} decomposes accordingly, as
\begin{align} 
\frac{\partial}{\partial t}
    \begin{bmatrix}
        {\sf m}_- \\
        {\sf m}_+ \\
        {\sf v}_- \\
        {\sf v}_+
    \end{bmatrix}
=
\begin{bmatrix}
    {\rm ad}_{{\sf m}_-} &{\rm ad}_{{\sf m}_+}& \partial_s + {\rm ad}_{{\sf v}_-} &  {\rm ad}_{{\sf v}_+} \\
    {\rm ad}_{{\sf m}_+} & {\rm ad}_{{\sf m}_-}& {\rm ad}_{{\sf v}_+} & \partial_s + {\rm ad}_{{\sf v}_-}\\
   \partial_s + {\rm ad}_{{\sf v}_-} &{\rm ad}_{{\sf v}_+}& 0& 0\\
   {\rm ad}_{{\sf v}_+}& \partial_s + {\rm ad}_{{\sf v}_-} &0 & 0
    \end{bmatrix}
    \begin{bmatrix}
   \delta h/\delta{{\sf m}_-} \\
   \delta h/\delta{{\sf m}_+} \\
   \delta h/\delta{{\sf v}_-} \\
   \delta h/\delta{{\sf v}_+} 
    \end{bmatrix}\, .
\end{align}
As seen in the previous section, the Hamiltonian cannot depend on $\mathfrak k$ in order to obtain an equation on a symmetric space, while writing it on the Lie algebra $\mathfrak g$. 
The previous system thus simplifies to 
\begin{align} 
    \begin{split}
\frac{\partial}{\partial t}
    \begin{bmatrix}
        {{\sf m}_+} \\
        {{\sf v}_+}
    \end{bmatrix}
    &=
    \begin{bmatrix}
        {\rm ad}_{{\sf m}_-} &  \partial_s + {\rm ad}_{{\sf v}_-}\\
        \partial_s + {\rm ad}_{{\sf v}_-} &0 
    \end{bmatrix}
    \begin{bmatrix}
       \delta h/\delta{{\sf m}_+} \\
       \delta h/\delta{{\sf v}_+} 
    \end{bmatrix}\\
    \frac{\partial}{\partial t}
    \begin{bmatrix}
        {{\sf m}_-} \\
        {{\sf v}_-} \\
    \end{bmatrix}
    &=
    \begin{bmatrix}
       {\rm ad}_ {{\sf m}_+}&  {\rm ad}_{{\sf v}_+} \\
       {\rm ad}_{{\sf v}_+}&0
    \end{bmatrix}
    \begin{bmatrix}
       \delta h/\delta{{\sf m}_+} \\
       \delta h/\delta{{\sf v}_+} 
    \end{bmatrix}\, .
    \end{split}
    \label{symspace-system}
\end{align}
This system reflects the structure of symmetric spaces. Namely, the $+$ variables are advected by the $-$ variables and the evolution of the $-$ variables only depends on the $+$ variables. 

\subsection{Reduction to an integrable $\sigma$-model}

Although the system \eqref{symspace-system} is rather general, it can be reduced to an integrable $\sigma$-model by using a quadratic Hamiltonian which depends only on the symmetric space variables indexed by $+$ in accordance with the discussion of the section \ref{un-reduction}. 
For selected constants $a$ and $b$, we set 
\begin{align}
        h({\sf m_+}, {\sf v_+} ) = \frac12 \int \left (a\|{\sf m_+}\|^2 + b\|{\sf v_+}\|^2  \right) ds\, , 
\end{align}
so that 
\begin{align}
    \frac{\delta h}{\delta{{\sf m}_{+}}} =a {\sf m}_{+} \quad\mathrm{and}\quad \frac{\delta h}{\delta{{\sf v}_{+}}} = b{\sf v}_{+}\, . 
\end{align}
Notice that if the Hamiltonian had included the $\sf m_-$ term, one would still have obtained the equation $(\sf m_-)_t=0$, so that we can set $\sf m_-=0$. 
The Hamiltonian structure  for the Hamiltonian simplifies to 
\begin{align} \label{LP-Ham-struct-reduced}
\frac{\partial}{\partial t}
    \begin{bmatrix}
    {{\sf m}_+} \\
    {{\sf v}_+} \\
    {{\sf v}_-}
    \end{bmatrix}
=
\begin{bmatrix}
    {\rm ad}_{{\sf m}_- }&\partial_s + {\rm ad}_{{\sf v}_-} & {\rm ad}_{{\sf v}_+} \\
  \partial_s + {\rm ad}_{{\sf v}_-}& 0& 0 \\
  {\rm ad}_{{\sf v}_+} & 0 &0
    \end{bmatrix}
    \begin{bmatrix}
    \delta h/\delta{{\sf m}_+} \\
    \delta h/\delta{{\sf v}_+} \\
    \delta h/\delta{{\sf v}_-}=0 \\
    \end{bmatrix}\, ,
\end{align}
and we arrive at the following system of equations,
\begin{align} \label{newsigma}
        \begin{split}
        \partial _t {\sf m}_+  &= \partial _s {\sf v}_++ \left[ {\sf v}_-,{\sf v}_+\right] \,,\\
        \partial _t{\sf v}_+&= \partial _s {\sf m}_+ +\left[ {\sf v}_-,{\sf m}_+\right] \,,\\
        \partial _t {\sf v}_-&=\left[ {\sf v}_+,{\sf m}_+\right] \, .
        \end{split}
\end{align}
This is the reduction of the $G$-strand equation on symmetric space to an integrable $\sigma$-model with a zero curvature relation $L_t-M_s+[L,M]=0$, given by the operators
\begin{align}
    \begin{split}
    L&=- {\sf v}_-  -\frac{\lambda}{2}({\sf v}_+ +{\sf m}_+)-\frac{1}{2\lambda}({\sf v}_+ -{\sf m}_+) \\
    M&=-\frac{\lambda}{2}({\sf v}_+ +{\sf m}_+)+\frac{1}{2\lambda}({\sf v}_+ -{\sf m}_+)\, .
    \end{split}
\end{align}
We refer to \cite{Se96, SeSe} for another derivation of these equations. 

We can apply the same change of variables as for the classical chiral model to recast this system into a more familiar form. 
We use $t= \frac12 (x-y) $, $s= \frac12 (x+y)$ and $\xi= m_+-v_+, \eta= m_++v_+$ to obtain
\begin{align}
        \begin{split}
        \partial_x \xi&= \partial _y \eta + \left[ {\sf v}_-,\eta-\xi\right] \\
        \partial_x\xi&= -\partial _y \eta -\left[ {\sf v}_-,\eta+\xi\right] \\
        (\partial_x&-\partial_y) {\sf v}_-=\frac12[\eta,\xi]\,  .
        \end{split}
\end{align}
This system can then be written as 
\begin{align}
        \begin{split}
                \partial_x \xi &= -[v_-,\xi]\\  
                \partial_y \eta &= -[v_-,\eta]\\        
                (\partial_x-\partial_y) {\sf v}_-&=\frac12[\eta,\xi]\, ,
        \end{split}
\end{align}
and the two fields of the relation \eqref{zero-curv1} become 
\begin{align}
\begin{split}
    Q&=- {\sf v}_-  -\lambda\eta \qquad P=-{\sf v}_- +\frac{1}{\lambda}\xi\, ,
\end{split}
\end{align}
with a similar zero curvature relation, i.e. $P_x-Q_y+[P,Q]=0$. 
This equation is clearly different from the $\mathfrak{su}(2)$ model. 
In the case $t\to it$ we have $x\equiv z = s+it,$ $y=\overline{z}=s-it$. 
This example was studied in \cite{Gu} and \cite{Harnad1984} where it is shown that the solutions can be constructed via the dressing method.

\subsection{Example: $\mathfrak{su}(2)$ chiral model on symmetric spaces}

When $\mathfrak g$ is $\mathfrak{su}(2)$ and $\mathfrak k$ is its Cartan subalgebra we have 
\begin{align}
        {\sf v}_-=ia\, \sigma_3 \in \mathfrak k, \quad\mathrm{where}\quad  \sigma_3=\text{diag} (1, -1)\, ,
\end{align} 
and $a(s,t)$ is a real scalar function.
We express ${\sf v}_+$ and ${\sf m}_+$ in function of two complex fields $A_1(s,t)$ and $A_2(s,t)$ as
\begin{align}
    \begin{split}
 {\sf v}_+ & =\left( \begin{array}{cc}
   0 &  A_1+A_2  \\
  -(\overline{A}_1 +\overline{A}_2) &  0 \\
\end{array}  \right) \, \in  \mathfrak p\, ,\qquad  
 {\sf m}_+  =\left( \begin{array}{cc}
   0 &  A_1-A_2  \\
  -(\overline{A}_1 -\overline{A}_2) &  0 \\
\end{array}  \right) \, \in  \mathfrak p\, ,
    \end{split}
\end{align}
where the bar denotes complex conjugation. Introducing a two dimensional complex vector ${\bf A}(s,t)=(A_1, A_2)^T$ and the two-dimensional cross-product ${\bf A}\times{\bf B}=A_1B_2-A_2B_1,$ the $G$-strand equations \eqref{newsigma} become 
\begin{align}
\begin{split}
    \sigma_3 {\bf A} _t& = {\bf A} _s + 2ia {\bf A} \qquad  i a_t = 2 {\bf A} \times \overline{\bf A}\, .
\end{split}
\end{align}
Note that the second equation can expressed as $a_t = - 4\Im(A_1\overline A_2)$, or $a_{tt} = -4\Im ( \sigma_3\mathbf{A}_s\times \overline{\mathbf{ A}})$, after taking one more time derivative.

The conserved quantities can be found by computing $\int \mathrm{Tr}(LM) ds$ from the operators in the compatibility relation \eqref{zero-curv1} and by looking separately at each term in the expansion in the spectral parameter $\lambda$. 
Only two terms do not vanish and give the conserved quantities 
\begin{align}
        C_1= \int |A_1|^2 ds \quad \mathrm{and} \quad C_2= \int |A_2|^2  ds\, .
\end{align}
Notice that while $C_1+C_2$ is the Hamiltonian, $C_1$ and $C_2$ are in fact conserved individually. 
Their associated continuity equations are 
\begin{align}
        \partial_t |A_1|^2 = \partial_s |A_1|^2 \quad \mathrm{and} \quad \partial_t |A_2|^2= - \partial_s |A_2|^2\,, 
\end{align}
or, in the $2$d vectorial form
\begin{align}
        \partial_t\left ( \mathbf{A}\overline{\mathbf{A}}\right ) = \partial_s \left (\mathbf A\sigma_3 \overline {\mathbf A  }\right )
        \quad \mathrm{and} \quad 
        \partial_t\left ( \mathbf{A}\sigma_3 \overline{\mathbf{A}}\right ) = \partial_s \left (\mathbf A\overline {\mathbf A  }\right )\, ,
\end{align}
illustrating the underlying covariance of the equations. 
The real form of this example is equivalent to the case $SO(3)/SO(2)$ by using the isomorphism between $SU(2)$ and $SO(3)$,  commonly used in the theory of complex fluids.

\subsection{Another example: a chiral model on $SO(4)/SO(3)$}

We now increase the dimensions by studying an example with a semi-simple algebra of rank $2$. For this, we pick the symmetric space $SO(4)/SO(3)$. 
We go directly to the equations of motion, by using the previous theory for the general chiral model equations \eqref{newsigma}.
The Lie bracket of $\mathfrak{so}(4)$ is of dimension $6$ and can be written in term of vectors $(X,Y)\in \mathbb R^6$ and $(X',Y')\in \mathbb R^6$ as
\begin{align}
        [(X,Y),(X',Y')] = \left ( X\times X' + Y\times Y', X\times Y' + Y\times X'\right )\, ,     
    \label{so4-algebra}
\end{align}
and one can choose the first $\mathbb R^3$ for $\mathfrak k$. 

Although $\mathfrak{so}(4)$ may be decomposed as $\mathfrak{so}(4)=\mathfrak{so}(3)\oplus \mathfrak{so}(3)$ into a direct sum of subalgebras, this is not the Cartan decomposition of $\mathfrak{so}(4)$; since the complimentary space $\mathfrak{p}$ in the Cartan decomposition is not a subalgebra. The direct sum  decomposition corresponds to the disentangled Lie algebra decomposition of $\mathfrak{so}(4)$, whereas \eqref{so4-algebra} corresponds to the entangled Cartan decomposition $ \mathfrak{so}(4)=\mathfrak{so}(3)\oplus \mathfrak{p} $ with $ \mathfrak{p}\neq \mathfrak{so}(3)$ where $\mathfrak{p} \supset Y, Y'$ is a $3$-dimensional linear space. The two decompositions are associated to different $\mathfrak{so}(4)$ Lie algebra bases which are related via rotations in the $\mathfrak{so}(4)$ Lie algebra vector space. 
The explicit decomposition $(X,Y)$ from \eqref{so4-algebra} in matrix form is represented as
\begin{align*}
    \begin{bmatrix}
        0 & Y \\
        -Y^T & \mathfrak{so}(3) 
     \end{bmatrix} \, , 
     \label{so4/so3-matrix}
\end{align*}
 where the 3D row vector $Y \in \mathfrak{p}$, and the subalgebra $\mathfrak{k}=\mathfrak{so}(3)$ is parametrised by the components of the 3D vector $X$ in a standard way (via the ``hat map'', see \cite{Ho2011GM2,Ho-Iv2}).

\begin{remark}[The $SO(p+q)/(SO(p)\times SO(q)$ decomposition]
This construction is similar to the decomposition of the Lie algebra $\mathfrak{so}(p+q)$ and the derivation of $G$-Strand equation could be done for this algebra in a similar way.  
The algebra $\mathfrak{so}(p+q)$ contains matrices of the form 
\begin{align*}
        X= \begin{bmatrix}
                X_1 & X_2 \\
                -X_2^T & X_3 
        \end{bmatrix}\, , 
\end{align*}
where $X_1=-X_1^T \in \mathfrak{so}(p)$  is $p\times p$ matrix, $X_3=-X_3^T \in \mathfrak{so}(q)$ is $q\times q$, and $X_2$ is $p\times q$. 
If $S=(I_p, -I_q)$ with $I_n$ being an identity matrix $n\times n,$ the Cartan involution is $\tilde{\varphi}(X)=SXS^{-1}$, then $\mathfrak{k}=\mathfrak{so}(p)\times \mathfrak{so}(q)$, the eigenspace of $\tilde{\varphi}$ with eigenvalue $1$, contains matrices of the form 
\begin{align*}
        \begin{bmatrix}
        X_1 & 0 \\
        0 & X_3 
        \end{bmatrix}\, , 
\end{align*}
and for $SO(p+q)/(SO(p)\times SO(q))$, the complimentary space $\mathfrak{p}$ consists of matrices of the form
\begin{align*}
        \begin{bmatrix}
                0 & X_2 \\
                -X_2^T &0 
        \end{bmatrix}\,.
\end{align*}
If $p=1$ then $X_2$ is a $q$ - dimensional vector. In the $SO(4)/SO(3)$ example, the symmetric space consists of $3$-dimensional vectors. 
A very good explanation is given in \cite{h}, see also \cite{KuFo} for similar constructions in the context of the nonlinear Schr\"odinger equation.  

\end{remark}

For the symmetric space $SO(4)/SO(4)$ we can choose $ {\sf m}_+=(0,Y)$, $ {\sf v}_+=(0,Z)$ and $ {\sf v}_-=(X,0)$, the equations become
\begin{align}
    \begin{split}
    \partial_t(0,Y)&=\partial_s(0,Z)+[(X,0),(0,Z)]\,,\\%=(0,Z_s)+(0, X\times Z) \\
    \partial_t(0,Z)&=\partial_s(0,Y)+[(X,0),(0,Y)]\,,\\%=(0,Y_s)+(0, X\times Y) \\
    \partial_t(X,0)&=[(0,Z),(0,Y)]\, ,% =(Z\times Y, 0),
    \end{split}
\end{align}
or, in terms of $X,Y,Z$ only, 
\begin{align}
    \begin{split}
    Y_t &=Z_s+X\times Z \,,\\
    Z_t &=Y_s + X\times Y \,,\\
    X_t&=Z\times Y\, .
    \end{split}
\end{align}
This equation has a conserved quantity
\begin{align}
    C_1= \int Y\cdot Z\,ds\, ,
\end{align}
since 
\begin{align}
    \partial_t C_1 = \partial_s C_1\, .
    \label{C1-conservation}
\end{align}
The Hamiltonian $H= Z^2 + Y^2$ gives the conservative form $\partial_t H = \partial_s C_1$, that has the same flux $C_1$ as for  \eqref{C1-conservation}. 
The complete integrability is clear from the previous discussion, but nevertheless, let's rewrite the Lax pair, as it can be directly written on $\mathfrak{so}(3)$, with 
\begin{align}
    \begin{split}
    L&= - X - \frac{\lambda}{2} ( Z+ Y) - \frac{1}{2\lambda} ( Z- Y)\,,\\
    M&= -\frac{\lambda}{2}(Z+Y) + \frac{1}{2\lambda}( Z-Y)\, ,
    \end{split}
    \label{Lax-SO(4)}
\end{align}
where the commutator of the ZCR \eqref{zero-curv1} is the cross product. 
Despite the complete integrability of this system, we observe instabilities for high frequencies. 
This is found from the dispersion relation of the linearised equation around the equilibrium solution $X_e= xe_1, Y_e= ye_1$ and $Z_e= ze_1$, with $x,y,x \in \mathbb R$, by
\begin{align}
     \left( \kappa^2 - \omega^2\right)  \left(  (\omega^{2}+z^2)^2 - (\kappa^{2}- (x^2+y^2))^2\right ) = 0 \, .
\end{align}
Apart from $\omega(\kappa)= \pm \kappa$, there are $4$ other branches
\begin{align}
    \omega(\kappa) =\pm \sqrt{-z^2 \pm(  \kappa^{2} -  (x^{2} + y^{2}))}\, ,
\label{omega(k)}
\end{align}
one of which shows instabilities at high wave numbers, similarly to the $G$-strand equations derived in \cite{FDT}. 

Following the approach of \cite{arnaudon2016deformation} for deforming integrable systems using the Sobolev $H^1$ norm, we may introduce nonlocality into the $G$-strand equations in order to regularise them.
We use the notation $\Lambda L:=(1-\alpha^2\partial_s^2)L$, and a direct application of \cite{arnaudon2016deformation} shows that $\Lambda$ does not appear in the $M$ operator. 
We, therefore, obtain the following deformed equations of motion
\begin{align}
    \begin{split}
    \Lambda Y_t &=Z_s+(\Lambda X)\times Z \\
    \Lambda Z_t &=Y_s + (\Lambda X)\times Y \\
    \Lambda X_t&=\frac12 \left ( (\Lambda Z)\times Y+ Z\times (\Lambda Y)\right ) \, .
    \end{split}
\end{align}
This equation seems not to be integrable, as the $\lambda^2$ and $\lambda^{-2}$ terms do not vanish.
Interestingly, the corresponding dispersion relation becomes
\begin{align}
    \omega(\kappa) =\pm \sqrt{ -z^2\pm \frac{ \kappa^{4}    - (x^{2}+A(\alpha) y^2)^2 }{(1+\alpha^{2} \kappa^{2})^2}}\, .
\end{align}
The unstable branch is thus bounded from above by $\frac{1}{\alpha^2}$ and the ill-posedness of the original equation is replaced by an unstable wave number regime. 
Other regularising terms may be added, or the limit $\alpha\to \infty$ may be taken in order to reduce these instabilities. 
In the latter case, the upper bound for the unstable branch will tend to $0$. 
Notice that in this limit, the $\partial_s$ terms will disappear, and the system will be equivalent to the finite dimensional reduction of the system explored below, which is integrable. 

{\bf Finite dimensional reduction.} In the case of $s$-independent fields, the system reduces to an integrable nine dimensional dynamical system given in vector form by
\begin{align}
    \begin{split}
            X_t&=Z\times Y, \quad Y_t =X\times Z \quad \mathrm{and} \quad Z_t =X\times Y\, .
    \end{split}
    \label{XYZ}
\end{align}
The Hamiltonian is $H=\frac12(\|Y\|^2+\|Z\|^2)$ and there are three conserved quantities
\begin{align}
\begin{split}
    C_1&= Y\cdot Z,\quad  C_2 = Y\cdot X,\quad C_3=  X\cdot Z,\
    \quad \hbox{and a Casimir} \quad C_4 =\|X\|^2+\|Z\|^2\, . 
\end{split}
\end{align}
Only the last conserved quantity is a Casimir, as one can see from the form of the Hamiltonian structure $J:\mathbb R^9\times \mathbb R^9\to \mathbb R^9$ given by
\begin{align}
    J(X,Y,Z)\nabla f := 
    \begin{pmatrix}
        0 & Z\times & 0 \\
        Z\times &0  & X\times \\
        0 & X\times & 0 
    \end{pmatrix}
    \begin{pmatrix}
        \partial_X f\\
        \partial_Y f\\
        \partial_Z f
    \end{pmatrix}\, .
    \label{J-XYZ}
\end{align}
The system thus lives in a eight dimensional space given as a subspace of $\mathbb R^8$ with Casimir $C_4$.
The remaining four constants of motions give us the complete integrability, provided they are in involution with respect to the Hamiltonian structure $J$ in \eqref{J-XYZ}. 
The latter fact can directly be checked and these conserved quantities can also be computed by expanding the quantity $ M\cdot L$ of \eqref{Lax-SO(4)} in term of powers of $\lambda$. 
Notice that these are not all conserved for the $1+1$ equation, as the complete $M$ operator is an infinite series in $\lambda$, which we will not compute here. 

Explicit solutions may be computed in term of elliptic curves directly using the Lax pair, but we leave this exercise for elsewhere and illustrates typical orbits via numerical integration in Figure \ref{fig:Fig-XYZ}.
Most of the orbits are periodic, and we displayed one which has intersecting orbits and another one which has degenerate periodic orbits. 
\begin{figure}[ht]
    \subfigure{\includegraphics[scale=0.45]{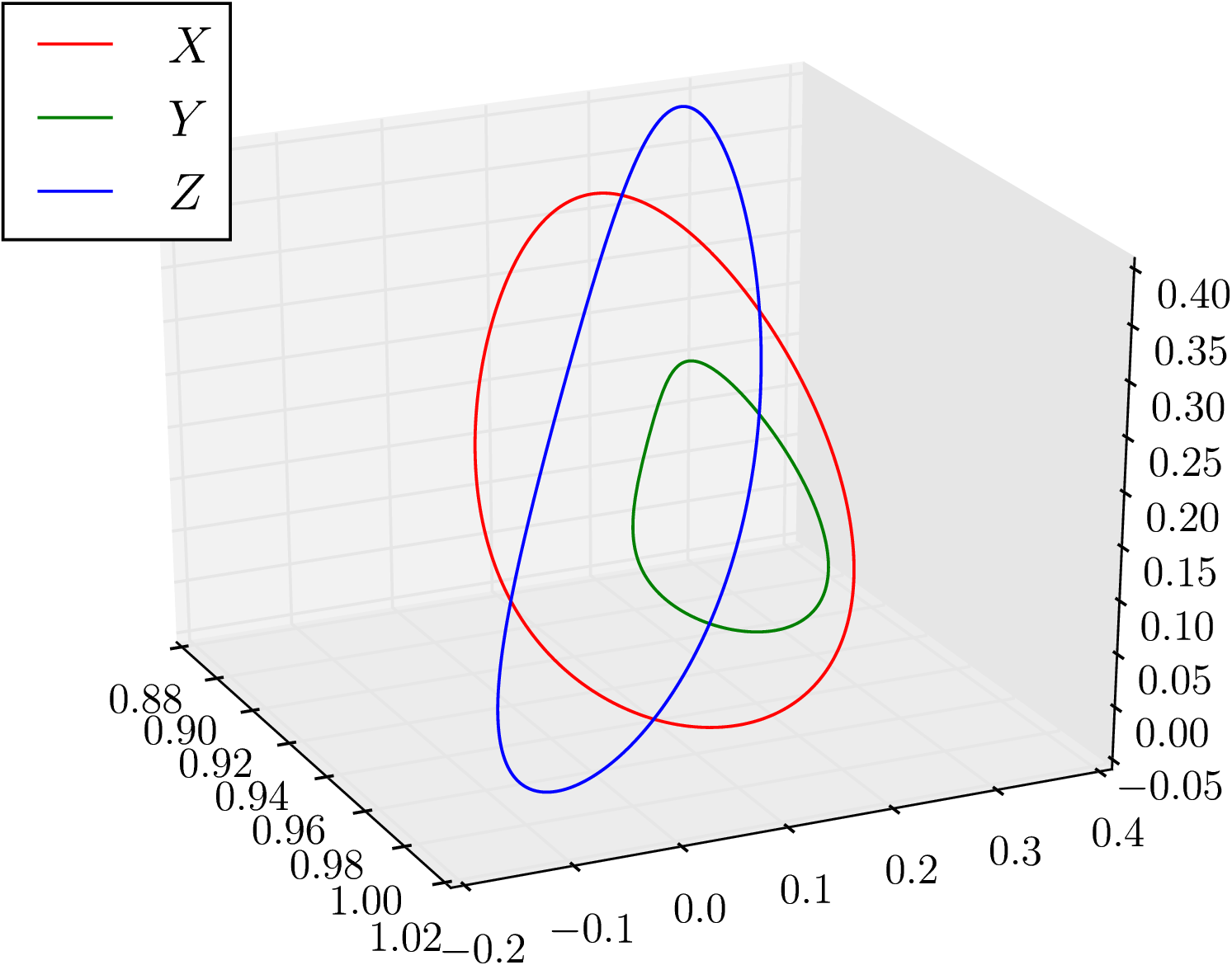}}
    \subfigure{\includegraphics[scale=0.45]{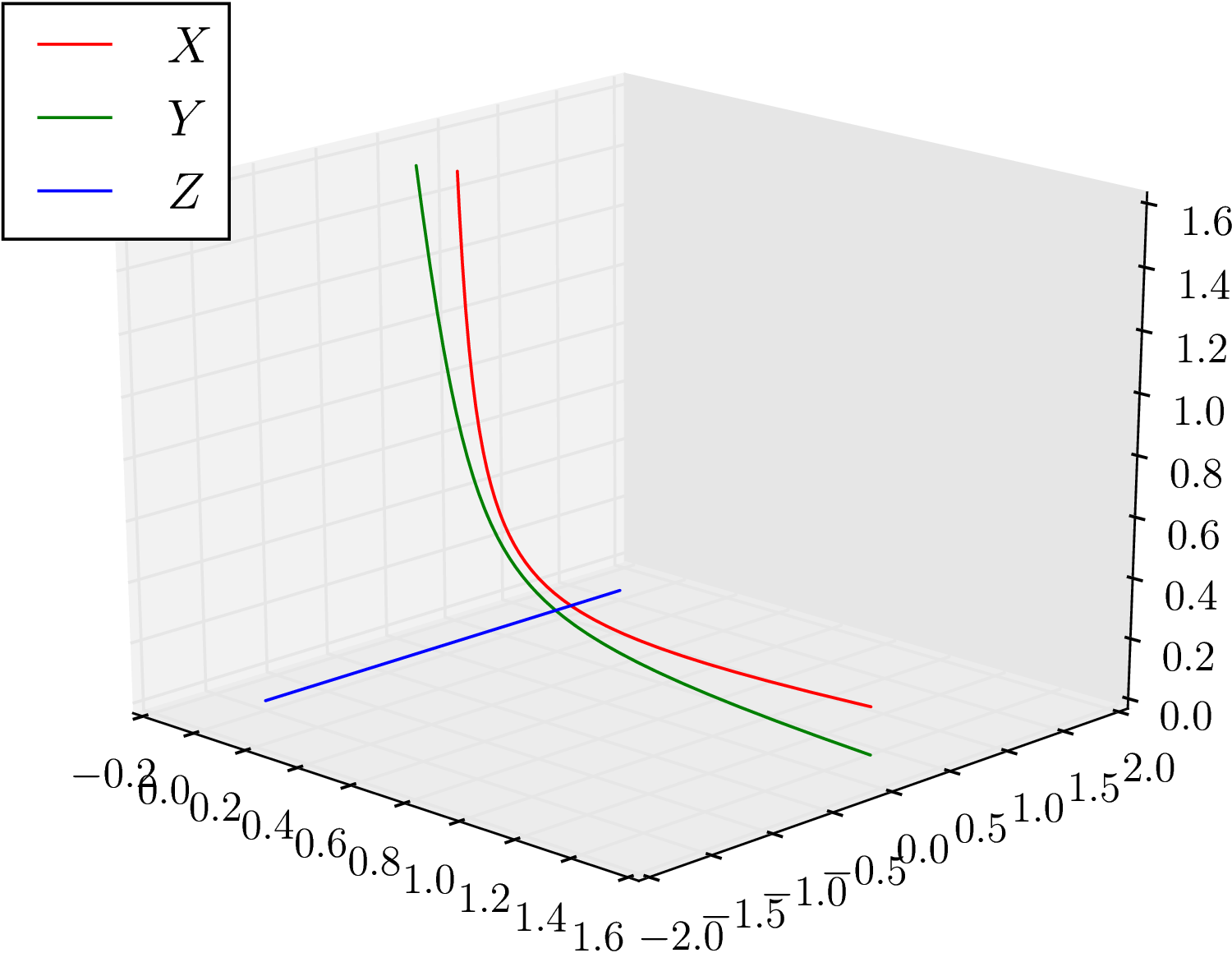}}
    \caption{We display here periodic orbits of the system \eqref{XYZ} arising from two different initial conditions. 
Both cases show interesting features. The left panel shows orbits of $X$ and $Z$ orbits intersecting in four points, and the right one shows degenerate periodic orbits.  }
    \label{fig:Fig-XYZ}
\end{figure}

The equilibrium solutions of this system are given whenever a linear combination of $X$, $Y$ and $Z$ vanishes. In this case, the quantity $H_C:=H+\sum_{i=1}^4\mu_i C_i$ has a critical point so that $\delta H_C=0$ for a particular set of coefficients $\mu_i$. 
Thus, the equilibria occur when $X$,$Y$ and $Z$ are aligned. 
These equilibria are stable provided the coefficients $\lambda_i $ are chosen so that the second variation $\delta^2 H_C$ has a strictly definite signature.

\section{ $\mathrm{ Diff}(\mathbb R)$-strands on symmetric spaces} \label{diff-strand}

We now turn to the study of $G$-strands on the diffeomorphism group, already investigated in \cite{Ho-Iv1}, but not in the context of symmetric spaces. 
The symmetric space structure for the diffeomorphism groups corresponds to the even or odd functions, thus the corresponding Diff-strand equations will show a particular interaction between the odd and even parts of the functions. 
We will start by recalling previous results on Diff-strands, and derive the equations with the symmetric space structure, then end with an example of strand peakon anti-peakon collisions. 

\subsection{Camassa-Holm equation on symmetric spaces}

We first recall the Camassa-Holm equation \cite{Ca-Ho}
\begin{align}
        m_t+2mu_x+m_xu=0,\qquad m=u-u_{xx}\,, 
    \label{CH_1}
\end{align}
which can be written in a Hamiltonian form
\begin{align}
     m_t=\{m,H_1\}\, , 
    \label{CHPB}
\end{align}
where, assuming for convenience that $m$ is $2\pi$ periodic in $x$, i.e.
$m(x)=m(x+2\pi)$, the Poisson bracket is 
\begin{align}
    \{F,G\}&:= -\int_{0} ^{2 \pi}\frac{\delta F}{\delta m}\left(m \partial+\partial m \right)\frac{\delta G}{\delta m}\,{\text d}x = -\int_{0} ^{2 \pi}\frac{\delta F}{\delta m}\mathrm{ad}^*_{\frac{\delta G}{\delta m}}m \,{\text d}x \, , 
    \label{PB_1}
\end{align}
and the Hamiltonian is
\begin{align}
    H_1(m)=\frac{1}{2}\int_{0} ^{2 \pi} m\left (1-\partial^2\right )^{-1} m\, {\text d}x\, . 
    \label{H_1}
\end{align}
Although the CH equation \eqref{CH_1} is integrable and admits a bi-Hamiltonian structure, we will not discuss its integrability here. Instead, we will exhibit its symmetric space structure by relating the Poisson bracket \eqref{PB_1} to the Witt algebra $W$. We refer the interested reader to \cite{Iv2005} and references therein for more details on this topic and in particular when a central extension is used to produce the dispersive Camassa-Holm equation.  
We did not implement this extension here and leave it for future works. 
First, upon imposing periodic boundary conditions on solutions of the CH equation, we may use the Fourier series expansion of the solution $m$, written as  
\begin{align}
    m(x,t)&=\frac{1}{2\pi}\sum_{n\in\mathbb{Z}}L_n(t) e^{inx}\, . 
    \label{exp-m}
\end{align}
Notice that the reality of $m$ gives the relation $L_{-n}=\overline{L}_n$.
Then, the Fourier coefficients $L_{n}$ form a classical Witt algebra without a central charge with respect to the Poisson bracket \eqref{PB_1}.
This is seen by directly computing the Poisson brackets among the Fourier coefficients
\begin{align}
    i\{L_{n},L_{m}\}=(n-m)L_{n+m}\, . 
    \label{eq8} 
\end{align}
In the Witt algebra, $ H_0=\frac{L_0}{2\pi}$ is an integral of motion and the Hamiltonian \eqref{H_1} decomposes as 
\begin{align}  
        H_1= \frac{1}{4\pi}\sum_{n\in \mathbb{Z}}\frac{L_{n}L_{-n}}{1+n^2}\, .
    \label{ham} 
\end{align}

One may now recognise the structure of a symmetric space, as follows. As linear sub-spaces of the Witt algebra $W$, 
\begin{enumerate}
        \item $\frak{k}$ is spanned by $L_{2k}$, $k \in \mathbb{Z}$, and
        \item  $\frak{p}$ is spanned by $L_{2k+1}$; $k \in \mathbb{Z}$.
\end{enumerate}
Then the Witt algebra can be decomposed as a direct sum $W=\frak{k}\oplus \frak{p}$ and one may check that it satisfies the commutation relation of a symmetric space given in \eqref{sym-relation}. 
This decomposition is in fact obtained by a $\mathbb{Z}_2$--grading of the Witt algebra and is equivalent to the splitting of $m(x, \cdot)$ into even and odd functions. 
Indeed, the algebra of vector fields $\mathfrak{v}$ on the circle with commutator 
\begin{align}
        [f(x),g(x)]=fg_x - f_xg 
\end{align}
is isomorphic to the Witt algebra. 
This is obvious if one takes a basis $l_n=e^{inx}$, $n\in \mathbb{Z}$.
As an infinite dimensional vector space $\mathfrak{v}$ is the space of the $2\pi$ periodic functions, which can be expanded in Fourier series over $l_n$.
Now, the Cartan involution 
\begin{align*}
        ( \tilde{\varphi} m)(x)=m(x+\pi)\, , 
\end{align*}
splits $\mathfrak{v}$ into subspaces of even mode functions (where $\tilde{\varphi}$ has eigenvalue 1) and odd mode functions (where $\tilde{\varphi}$ has eigenvalue -1). 
Moreover, the decomposition into even and odd modes is 
\begin{align*}
        m^{em}(x)=(1/2)(m(x)+m(x+\pi))\quad \mathrm{and} \quad m^{om}(x)=(1/2)(m(x)-m(x+\pi))\, . 
\end{align*}
These are orthogonal with respect to the $L_2$-inner product
\begin{align}
        (m(x),n(x))=\frac{1}{2\pi} \int_{0}^{2\pi}m n dx \Longleftrightarrow (l_n, l_k)=\delta_{n+k,0}\, .
\end{align}
By a new variable identification $e^{ix}\mapsto z$ the subspace of the even modes is naturally isomorphic to the subspace of even functions, and the subspace of odd modes is isomorphic to the subspace of the odd functions. Thus, we can identify $\frak{k}$ with the subspace of even functions, and  $\frak{p}$ with the subspace of odd functions.

It is known that the CH equation admits odd solutions, including the peakon solutions, see \cite{constantin2000existence,constantin2002geometric} and references therein.
This will also be the case for the $G$-strand constructions, as we will see below.

\subsection{ Diff($\mathbb{R}$)-strand equations on symmetric spaces}

We now derive the equation of motion for the Diff($\mathbb R$)-strands where $s$ denotes the strand variable and the $x$ coordinate labels Diff($\mathbb R$), where the odd or even functions are defined.
For this, we consider a Lagrangian $\ell=\ell(u,v)$ depending on two fields $u(s,t,x)$ and $v(s,t,x)$.
We will also need to introduce the momenta $m={\delta\ell}/{\delta u } $ and $n={\delta\ell}/{\delta v }$.  
The equations arising are right-invariant ${\rm Diff}(\mathbb{R})$-strand equations for maps $\mathbb{R}\times\mathbb{R}\to G={\rm Diff}(\mathbb{R})$ and in one spatial dimension they may be expressed as a system of two 2+1 PDEs in $(s,x,t)$,
\begin{align}
    \begin{split}
    m_t + n_s &+ {\rm ad}^*_{u }m + {\rm ad}^*_{v }n =0  \\ 
    v_t - u _s &+ {\rm ad}_v u =0 \,,
    \end{split}
\label{Gstrand-eqn2R}
\end{align}
or, using the explicit form of the coadjoint action, 
\begin{align}
    \begin{split}
    m_t + n_s &+ u m_x + 2m u_x + v n_x +  2nv_x=0  \\
    v_t - u _s &+ u v_x - v u_x =0 \,.
    \end{split}
\label{Gstrand-eqn2R-explicit}
\end{align}
The Hamiltonian structure for these  ${\rm Diff}(\mathbb{R})$-strand equations is obtained by the Legendre transformation
\begin{align*}
    h(m,v)=\langle m,\, u\rangle_{L^2} - \ell(u,\,v)\, ,
\end{align*}
resulting in
\begin{align}
    \frac{d}{dt}
    \begin{bmatrix}
        m \\ 
        v
    \end{bmatrix}
    =
    \begin{bmatrix}
    -{\rm ad}^*_\square m  &\quad \partial_s + {\rm ad}^*_v \\
    \partial_s - {\rm ad}_v  &\quad 0
    \end{bmatrix}
    \begin{bmatrix}
        {\delta h}/{\delta m} = u \\
        {\delta h}/{\delta v} = - n
    \end{bmatrix}\, .
    \label{1stHamForm}
\end{align}
Using the odd/even decomposition introduced in the previous section, the Hamiltonian structure \eqref{1stHamForm} becomes
\begin{align} \label{LP-Ham-struct-sym space}
\frac{\partial}{\partial t}
    \begin{bmatrix}
         m^e \\
         m^o \\
         v^o \\
         v^e
    \end{bmatrix}
    =
    \begin{bmatrix}
     - {\rm ad}^\ast_\square { m^o}
       &-{\rm ad}^\ast_\square { m^e}&{\rm ad}^*_{ v^e} & \partial_s + {\rm ad}^*_{ v^o} \\
      - {\rm ad}^\ast_\square { m^e}
        &-{\rm ad}^\ast_\square { m^o}& \partial_s + {\rm ad}^*_{v^o} & {\rm ad}^*_{v^e} \\
      -{\rm ad}_{v^e} & \partial_s - {\rm ad}_{v^o} & 0 & 0 \\
       \partial_s - {\rm ad}_{v^o}& -{\rm ad}_{v^e} & 0 & 0
    \end{bmatrix}
    \begin{bmatrix}
       \delta h/\delta{ m^{e}=u^{e}} \\
       \delta h/\delta{ m^o=u^o} \\
       \delta h/\delta{{\sf v}^o=-n^o} \\
       \delta h/\delta{{\sf v}^e=-n^e} \\
    \end{bmatrix}\, .
\end{align}
In particular, when $h=h(m^o,v^o)$ depends only on the symmetric space variables, the odd and even parts of the equation decouple as
\begin{align} \label{odddyn}
    \frac{d}{dt}
    \begin{bmatrix}
        m^o \\ 
        v^o
    \end{bmatrix}
    &=
    \begin{bmatrix}
        -{\rm ad}^*_\square m^o  &\quad \partial_s + {\rm ad}^*_{v^o} \\
        \partial_s - {\rm ad}_{v^o}  &\quad 0
    \end{bmatrix}
    \begin{bmatrix}
        {\delta h}/{\delta m^o} = u^o \\
        {\delta h}/{\delta v^o} = - n^o
    \end{bmatrix}\\
    \frac{d}{dt} 
    \begin{bmatrix}
            m^e \\ 
        v^e
    \end{bmatrix}
    &=
    \begin{bmatrix}
            -\mathrm{ad}^*_\square m^e &  \mathrm{ad}^*_{v^e}\\
            -\mathrm{ad}_{v^e} & 0 
    \end{bmatrix}
    \begin{bmatrix}
        {\delta h}/{\delta m^o} = u^o \\
        {\delta h}/{\delta v^o} = - n^o
    \end{bmatrix}\, . 
    \label{evendyn}
\end{align}
The structure of these equations reflects the known property that only odd solutions can survive alone in the CH equation (which is the reduction $u\equiv v$, $m\equiv n$ and $s\equiv t$ of the ${\rm Diff}(\mathbb{R})$-strand equation) and are of the type peakon and anti-peakon collisions, governed by \eqref{odddyn}.
As soon as the solution has an even part, the dynamics become more complicated due to the coupling with \eqref{evendyn}. 
In addition, the second system for the even variables does not have any derivative with respect to the $s$ variables and the Hamiltonian structure depends only on the even variables. 

\subsection{Singular solutions of the ${\rm Diff}(\mathbb{R})$-strand  equations}

We now derive the explicit solution of the ${\rm Diff}(\mathbb{R})$-strand system by reduction to a system of peakons that still depend on two variables, $s$ and $t$. 
For simplicity, we will make the following choice for the Lagrangian
 \begin{align}
    \ell (u ,v ) = \frac12 \int (u_x^2+v_x^2)dx \,,
\label{Gstrand-pkn-Lag}
\end{align}
that corresponds to a two-component generalization of the Hunter-Saxton equation \cite{HS,HZ}.
The ${\rm Diff}(\mathbb{R})$-strand equations \eqref{Gstrand-eqn2R} admit peakon solutions in both momenta
\begin{align*}
    m=-u_{xx} \quad \text{ and} \quad n=-v_{xx}\, ,
\end{align*}
with continuous velocities $u$ and $v$. 
It can be directly checked that the ${\rm Diff}(\mathbb{R})$-strand equations \eqref{Gstrand-eqn2R} admit singular solutions expressible as a linear superposition of Dirac delta functions in the momenta. 
The general form is 
\begin{align}
\begin{split}
        m(s,t,x) &= \sum_a M_a(s,t)\delta(x-Q^a(s,t))\,, \qquad n(s,t,x) = \sum_a N_a(s,t)\delta(x-Q^a(s,t)) \\    
        u(s,t,x)  &=K*m=\sum_a M_a(s,t) K(x,Q^a)\, ,  \qquad v(s,t,x)  = K*n=\sum_a N_a(s,t) K(x,Q^a)\,,
    \end{split}
    \label{Gstrand-singsolns1}
\end{align}
where $K(x,y)= -\frac12 |x-y|$ is the Green function of the operator $-\partial_x^2$, i.e. $-\partial_x^2K(x,0)=\delta(x)$.
The solution parameters $\{Q^a(s,t), M_a(s,t), N_a(s,t)\}$ with $a\in\mathbb{Z}$ that specify the position and amplitude of singular solutions \eqref{Gstrand-singsolns1} are determined by the following set of evolutionary PDEs in $s$ and $t$, in which we denote $ K^{ab}:=K(Q^a,Q^b) $ with integer
summation indices $a,b,c,e\in\mathbb{Z}$:
\begin{align}
    \begin{split}
    \partial_t Q^a(s,t) &= u(Q^a,s,t) = \sum_b M_b(s,t) K^{ab} \\
    \partial_s Q^a(s,t) &= v(Q^a,s,t) =\sum_b N_b(s,t) K^{ab} \\
    \partial_t M_a(s,t) &= -\, \partial_s N_a -\sum_c (M_aM_c+N_aN_c) \frac{\partial K^{ac}}{\partial Q^a}
    \quad\hbox{(no sum on $a$)} \\
    \partial_t N_a(s,t) &=\partial_s M_a +    \sum_{b,c,e} (N_bM_c - M_bN_c) \frac{\partial K^{ec}}{\partial Q^e} (K^{eb}-K^{cb})(K^{-1})_{ae} \,.
    \end{split}
\label{Gstrand-eqns}
\end{align}
The last pair of equations in (\ref{Gstrand-eqns}) may be solved
as a system for the momenta
$(M_a,N_a)$, then used in the previous pair to update the positions $Q^a(t,s)$ of the singular solutions.
We will call these singular solutions `peakons' for simplicity, although the Green function, in this case, is unbounded and the shape is not the usual peakon shape.

\subsection{Example: Two-peakon solution of a ${\rm Diff}(\mathbb{R})$-strand}

We now study a simpler system of two `peakon' collisions and find explicit solutions.
If we denote the relative position of the two solutions by $X(s,t)=Q^1-Q^2$ we can express the
Green's function as $K=K(X)$ and the first two equations in \eqref{Gstrand-eqns} imply
\begin{align}
    \partial_t X = -(M_1-M_2)K(X)\, ,\qquad \partial_s X = - (N_1-N_2)K(X)\, .
    \label{Qdiff-eqns}
\end{align}
The second pair of equations in (\ref{Gstrand-eqns}) may then be
written as
\begin{align}
    \begin{split}
    \partial_t M_1 &= - \partial_s N_1 - (M_1M_2+ N_1N_2)K'(X) \\
    \partial_t M_2 &= - \partial_s N_2 + (M_1M_2+ N_1N_2)K'(X) \\
    \partial_t N_1 &=  \partial_s M_1 - (N_1M_2-M_1N_2)K'(X) \\
    \partial_t N_2 &=  \partial_s M_2 - (N_1M_2-M_1N_2)K'(X) \,.
    \end{split}
    \label{Gstrand-pp}
\end{align}
Assuming $X>0$ we have $K'(X)=-\frac{1}{2} \text{sgn}(X)=-\frac{1}{2}$ and introducing the variable $S_{1,2}=M_{1,2}+iN_{1,2}$ we can rewrite \eqref{Gstrand-pp} as  
\begin{align}
    (\partial_t -i\partial_s)S_1 = \frac{1}{2}S_1\overline{S}_2\, , \qquad (\partial_t -i\partial_s)S_2 = -\frac{1}{2}\overline{S}_1S_2 \,.
    \label{Gstrand-Lform}
\end{align}
The solution for $X$ can then be expressed formally via $S_{1,2}$ from
\eqref{Qdiff-eqns} as 
\begin{align*}
     X=\exp\left(\frac{1}{2}\Delta^{-1} \Re(S_1\overline{S}_2)\right)\, , 
\end{align*}
where $\Delta=\partial_t^2+ \partial_s^2$  and $\Re(z) $ is the real part of $z$.
From the system \eqref{Gstrand-Lform} we obtain  
\begin{align}
    \Delta \ln S_1 = -\frac{1}{4}S_1\overline{S}_2\quad \mathrm{and} \quad \Delta \ln S_2 = -\frac{1}{4}\overline{S}_1S_2\, , 
    \label{Gstrand-Lform3}
\end{align} 
thus $\Delta \ln S_1 = \Delta \ln \overline{ S}_2 $ and
$S_1=\overline{S}_2e^{h}$ where $h(s,t)$  is an arbitrary harmonic function, i.e. $\Delta h =0$. 
Then for the variable  $\widetilde{Y}=\ln S_1$ we have the equation 
\begin{align} 
    \Delta \widetilde{Y} = -\frac{1}{4}e^{2\widetilde{Y}-h}\, ,
    \label{Liouvile}
\end{align} 
and for $Y=\ln S_1-\frac{1}{2}h -2\ln 2+i \pi$ we arrive at the Liouville's $2D$ equation
\begin{align} 
    \Delta Y = e^{2Y}\, . 
    \label{Liouvile1}
\end{align}
Solutions of \eqref{Liouvile1} are known in the form 
\begin{align*}
     Y=\frac{1}{2}\ln\frac{w_s^2+w_t^2}{f(w)}\, ,
\end{align*}
where the function $f(w)$ can be either $w^2$, $\cos^2 w$, $\sin^2 w$ or $\sinh^2 w$ with
$w$ being an arbitrary harmonic function, i.e. $\Delta w=0$. We refer to \cite{DC,Ib,Ki02} for more details on the Liouville equation. 
From this computation, the solutions $S_{1,2}$ depend on two arbitrary complex harmonic functions $h,w$, hence the four peakon parameters $M_{1,2}$ and $N_{1,2}$ can be given in terms of four real arbitrary harmonic functions.
Although we lost track of the space-time symmetry of the original equations when we applied the Legendre transformation, some flavour of it still remains, in that the harmonic condition for the arbitrary functions $w$ and $h$ is symmetric in $s$ and $t$.

\subsection{Two-peakon solution on the symmetric space: the peakon-antipeakon solution}

The reduction of the system \eqref{Gstrand-pp} to a system with only odd functions of $x$ can be achieved by setting $Q_2=-Q_1$, $M_2=-M_1$, $N_2=-N_1$ thus $S_2=-S_1$ and $Q_1=X/2$, where
\begin{align} 
        X=\exp\left(-\frac{1}{2}\Delta^{-1}|S_1|^2\right)\, . 
    \label{X} 
\end{align}
We will solve this reduced system directly instead of starting from the general solution of the initial system.  

The equivalent equation \eqref{Gstrand-Lform} becomes, after imposing the symmetry condition, a single complex equation
\begin{align}
    (\partial_t -i\partial_s)S_1 = - \frac{1}{2}|S_1|^2\,.
    \label{redS}
\end{align}
The imaginary part of \eqref{redS} simplifies to
\begin{align}
    \partial_t N_1 = \partial_s M_1\, ,
    \label{redS1}
\end{align} 
and it can be solved by introducing a real scalar function $\psi$ such that 
\begin{align}
        M_1=\partial_t \psi \quad\mathrm{and} \quad  N_1= \partial_s \psi\, . 
    \label{MN1}
\end{align}
The real part of \eqref{redS} produces a dynamical equation for $\psi$
\begin{align}\label{psi}
    \Delta \psi +\frac{1}{2} \left( \psi_t^2 + \psi_s^2\right)=0\,  , 
\end{align}
or equivalently,
\begin{align}
    \Delta \exp \left(\frac{1}{2} \psi\right )=0\, .  
    \label{psi 1}
\end{align}
The solution is given in term of a real harmonic function $h(s,t)$ as
\begin{align}
    \psi=\ln (h^2(s,t))\, .
    \label{psi 2}
\end{align}
From \eqref{X} and \eqref{MN1} one can derive 
\begin{align}
    \Delta \ln |X| +\frac{1}{2} \left( \psi_t^2 + \psi_s^2\right)=0\, , 
    \label{|X|}
\end{align} 
and thus from \eqref{psi}
\begin{align}
    |X| = \exp (\psi(s,t))=h^2(s,t)\, .
    \label{solX}
 \end{align}  
From the definitions of $Q_1,M_1$ and $N_1$, we have 
\begin{align}
        Q_1=\frac{1}{2} h^2(s,t),  \qquad M_1(s,t)= \frac{2h_t}{h} \quad\mathrm{and}\quad  N_1(s,t)= \frac{2h_s}{h}\, .
    \label{solMN}
 \end{align}  
Similarly to the general case of the previous section, all solutions are parametrised by an arbitrary harmonic function $h(s,t)$. 
The solutions $m$, $n$ are obviously odd functions of $x$, as they can be written
\begin{align}
    \begin{split}
    m(s,t,x) &= \frac{2h_t}{h}\left( \delta\left(x-\frac{h^2}{2}\right)- \delta\left(x+\frac{h^2}{2}\right)\right) \\
    n(s,t,x) &= \frac{2h_s}{h}\left( \delta\left(x-\frac{h^2}{2}\right)- \delta\left(x+\frac{h^2}{2}\right)\right) \,. 
\end{split}
    \label{oddsol}
\end{align} 
We end this section with an illustration of these solutions, using the harmonic functions indexed by $k\in \mathbb Z$  
\begin{align}
    h_k(s,t)= r^k\cos(k\theta), \quad \mathrm{where}\quad  r= \sqrt{s^2+t^2} \quad\mathrm{and}\quad  \theta = \arctan\left (\frac{t}{s}\right)\, . 
    \label{harm-func}
\end{align}
The main properties of the peakon anti-peakons are recovered, namely that they vanish at the position of interaction, and then exchange their momenta. 

\begin{figure}[ht]
    \subfigure[$t=-10$]{\includegraphics[scale=0.4]{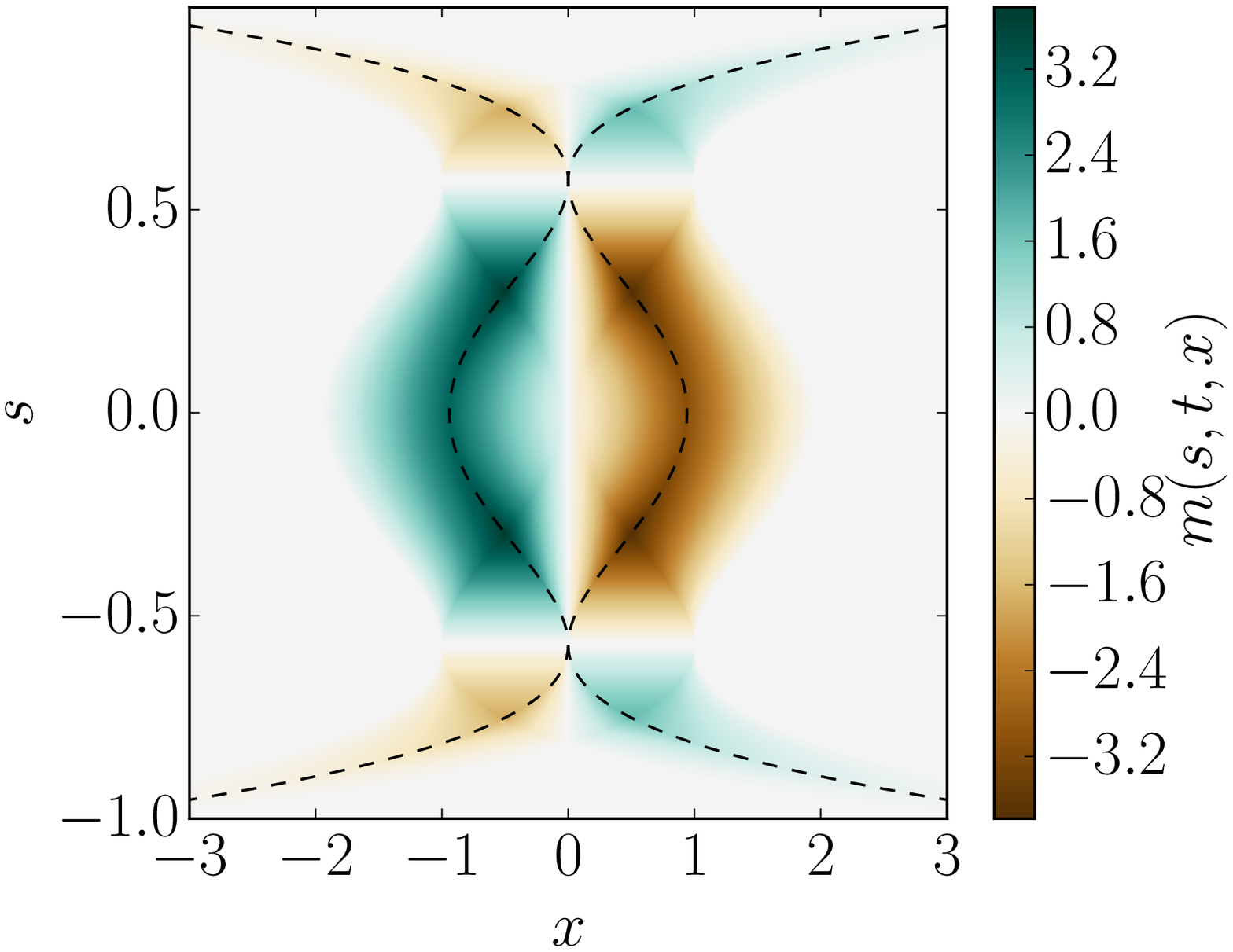}}
    \subfigure[$t=-7$]{\includegraphics[scale=0.4]{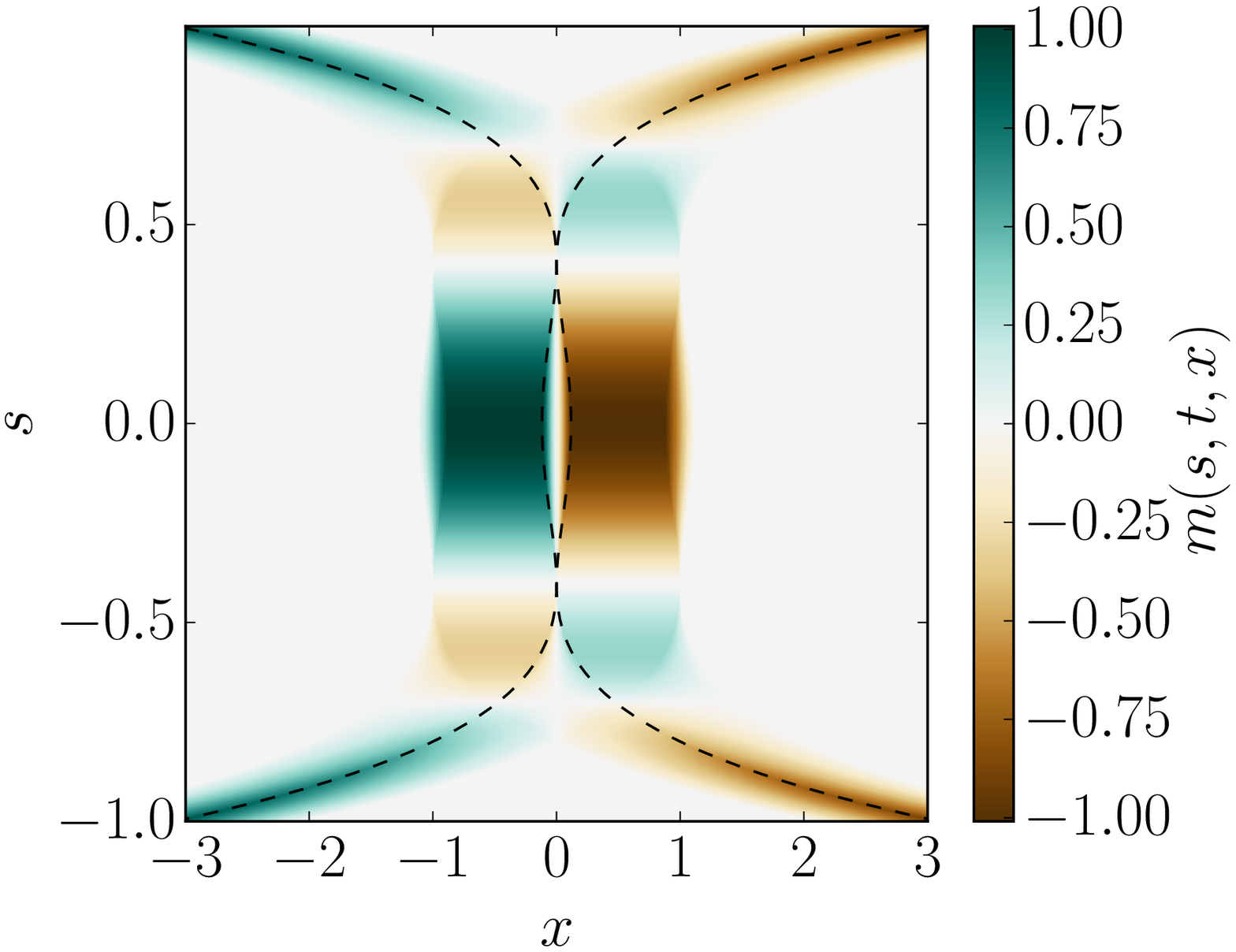}}
    \subfigure[$t=-0.1$]{\includegraphics[scale=0.4]{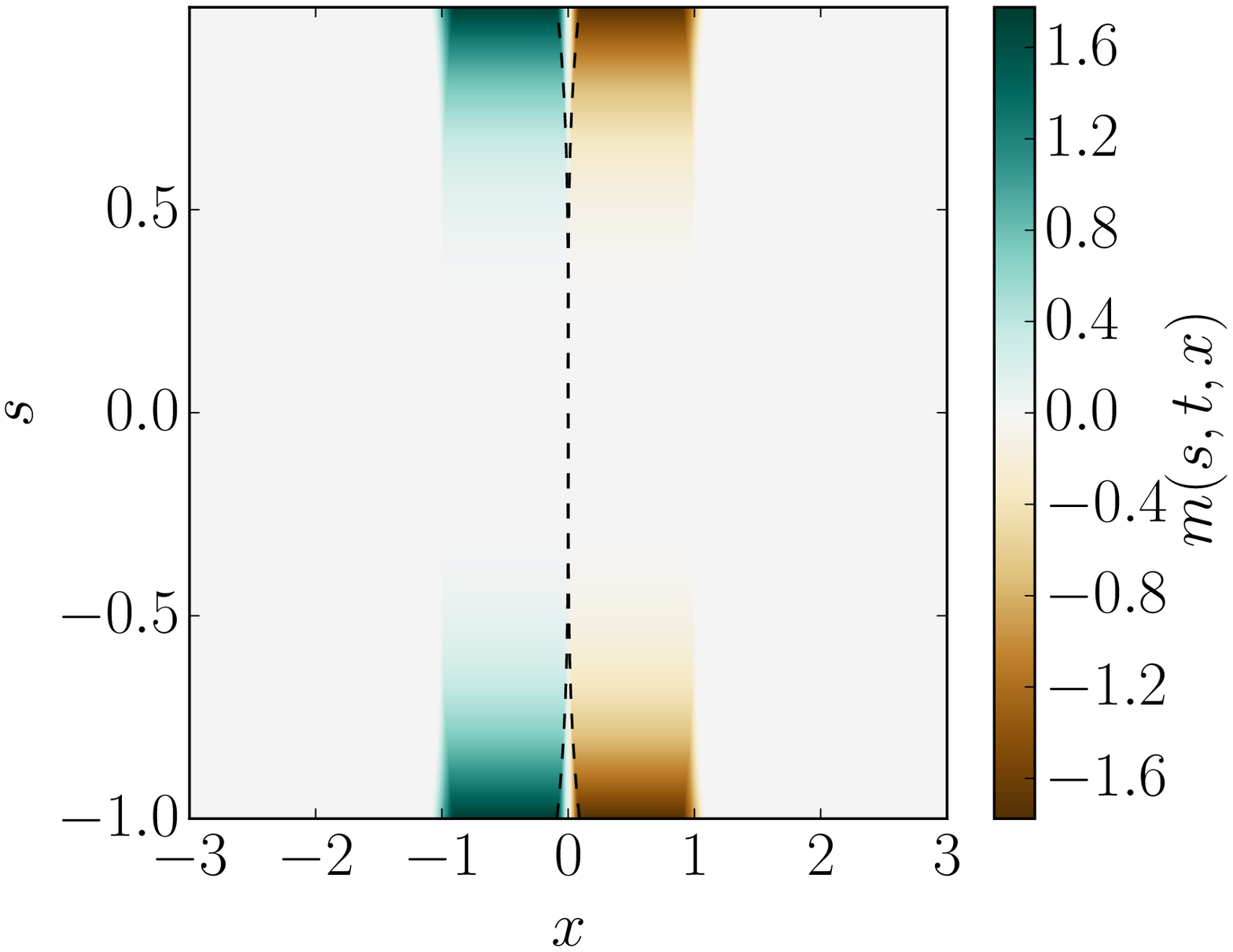}}
    \subfigure[$t=4$]{\includegraphics[scale=0.4]{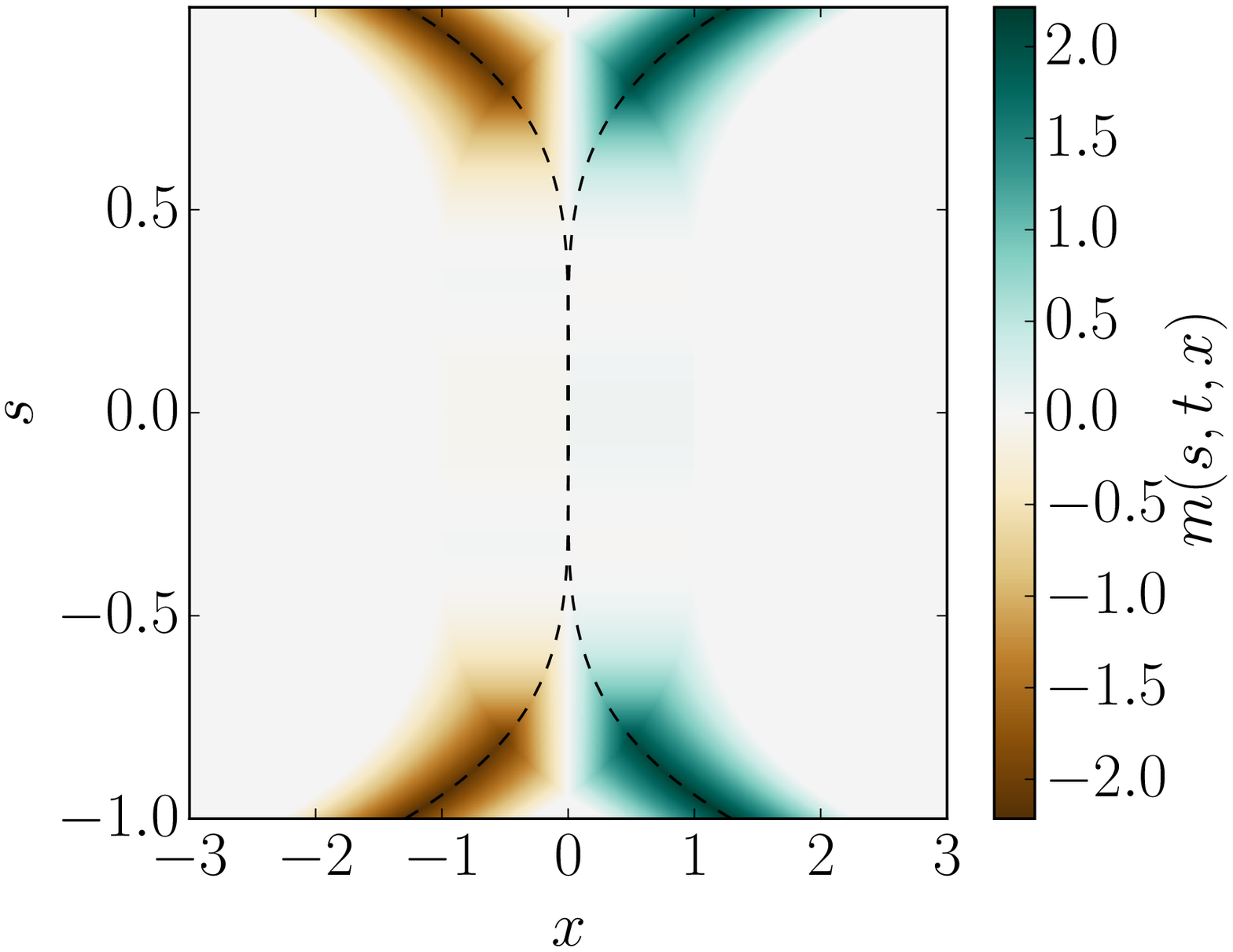}}
    \caption{We display the Hunter-Saxton peakon solution $m(s,t)$ in \eqref{oddsol} for the harmonic function \eqref{harm-func} with $k=3$ at various times around the collision at $t=0$. 
The order of the harmonic function gives the number of zeroes of the function; that is, the number of intersecting points of the two peakon strands.
    Note that the solution for $n(s,t)$ in \eqref{oddsol} is the same, thus not displayed here.}
    \label{fig:peakons}
\end{figure}

We display in figure \ref{fig:peakons} a few snapshots of the collision, obtain by plotting \eqref{oddsol} with the harmonic function \eqref{harm-func} with $k=3$
\footnote{See \url{http://wwwf.imperial.ac.uk/~aa10213/} for a video of the complete time evolution.}

\section{Conclusions}\label{conclusion-sec}

The $G$-strand PDEs arise from the Euler-Poincar\'e variational equations for a $G$-invariant Lagrangian, coupled to an auxiliary zero-curvature relation. 
The Hamiltonian formulation of the equations admits a natural phase space which takes values in the dual of a Lie algebra. 
This fact provides an opportunity for a further splitting of the phase space in a way consistent with the Lie-algebraic structures. 

In the present study, we have derived equations whose Hamiltonians depend on the variables from the complementary subspaces of a symmetric space. The examples included finite dimensional $G$-strands, as well as the Diff-strand example for the infinite-dimensional diffeomorphism group. The existence of odd solutions in the $x$ variable, which usually appear as peakon-antipeakon solutions of the Diff-strand equations, was shown to be rooted in the algebraic structure of the phase space, which can be orthogonally decomposed into subspaces of even and odd functions. The invariance of the complementary subspace for odd functions under special Hamiltonian flows arises because the phase space is a symmetric space. Examples in both finite and infinite dimensions considered here suggest topics for further studies. In particular, the finite dimensional integrable $XYZ$ system on the coset space $SO(4)/SO(3)$ in section \ref{symmetric-spaces} deserves further investigation to determine its full solution behaviour. Likewise, in the Diff-strand example in section \ref{diff-strand}, the remarkable reduction of the 2-peakon equations given by \eqref{Gstrand-Lform} to the peakon-antipeakon equation \eqref{redS} should raise additional interesting questions. 

\subsection*{Acknowledgements}
We are grateful for enlightening discussions of this material with F. Gay-Balmaz, T. S. Ratiu and C. Tronci.
We are also grateful for the careful and thoughtful reports of the referees.
Work by RII was supported by a {\it Seed funding} grant support from Dublin Institute of Technology for a project in association with the Environmental Sustainability and Health Institute,  Dublin. Work by DDH was partially supported by Advanced Grant 267382 FCCA from the European Research Council.
AA is supported by an Imperial College London Roth Award.

\bibliographystyle{alpha}
\bibliography{biblio.bib}

\end{document}